\documentclass[journal,draftclsnofoot,onecolumn,12pt]{IEEEtran}

\newif\ifarxiv
\arxivtrue 

\newif\ifonecolumn
\onecolumntrue 

\IEEEoverridecommandlockouts
\usepackage{multirow}
\usepackage{dsfont}
\usepackage{cite}
\usepackage{amsmath,amssymb,amsfonts, mathtools}
\usepackage{float}
\usepackage{makecell}
\usepackage{caption}
\usepackage{graphicx}
\usepackage{textcomp}
\usepackage{algorithm}
\usepackage[left=0.625in,right=0.625in,bottom=0.95in,top=0.75in]{geometry}
\usepackage{amsthm}
\usepackage{algpseudocode}
\usepackage{adjustbox}
\usepackage{xcolor}
\usepackage{verbatim}
\usepackage{paralist}
\usepackage{array}
\newcolumntype{?}{!{\vrule width 1.1pt}}
\usepackage{tablefootnote}

\newtheorem{proposition}[]{Proposition}

\theoremstyle{remark}

\definecolor{green}{rgb}{0.0, 0.5, 0.0} 
\allowdisplaybreaks
\usepackage{acro}
\newcolumntype{?}{!{\vrule width 1pt}}

\DeclareAcronym{snr}{
  short = SNR,
  long = signal-to-noise ratio,}
  \DeclareAcronym{pdf}{
	short = PDF,
	long = probability density function,}
  \DeclareAcronym{sar}{
  short = SAR,
  long = synthetic aperture radar,}

\DeclareAcronym{insar}{
  short = InSAR,
  long = interferometric synthetic aperture radar,}

\DeclareAcronym{ao}{
	short = {AO},
	long = {alternating optimization},
	long-plural-form = {alternating optimizations}
}
\DeclareAcronym{mimo}{
	short = {MIMO},
	long = {multiple-input multiple-output}
}

\DeclareAcronym{uav}{
        short = {UAV},
        long = {unmanned aerial vehicle},
        long-plural-form = {unmanned aerial vehicles}
}

\DeclareAcronym{fdma}{
	short = {FDMA},
	long = {frequency-division multiple-access},
}
\DeclareAcronym{1d}{
	short = {1D},
	long = {one-dimensional},
}
\DeclareAcronym{islr}{
	short = {ISLR},
	long = {integrated sidelobe ratio},
}

\DeclareAcronym{pslr}{
	short = {PSLR},
	long = {peak sidelobe ratio},
}

\DeclareAcronym{3d}{
        short = {3D},
        long = {three-dimensional},
}
\DeclareAcronym{pso}{
	short = {PSO},
	long = {particle swarm optimization},
}

\DeclareAcronym{2d}{
        short = {2D},
        long = {two-dimensional},
}

\DeclareAcronym{dem}{
        short = {DEM},
        long = {digital elevation model},
}
\DeclareAcronym{gs}{
        short = {GS},
        long = {ground station},
        long-plural-form = {ground stations}
}

\DeclareAcronym{los}{
        short = {LOS},
        long = {line-of-sight},
}

\DeclareAcronym{sca}{
        short = {SCA},
        long = {successive convex approximation},
}

\DeclareAcronym{nesz}{
        short = {NESZ},
        long = {noise equivalent sigma zero},
}

\DeclareAcronym{wrt}{
        short = {w.r.t.},
        long = {with respect to },
}
\DeclareAcronym{rhs}{
        short = {r.h.s},
        long = {right-hand side },
}
\DeclareAcronym{gmti}{
	short = {GMTI},
	long = {ground moving target indication},
}
\DeclareAcronym{lhs}{
        short = {l.h.s},
        long = {left-hand side },
}
\DeclareAcronym{bcd}{
	short = {BCD},
	long = {block coordinate descent},
}
\DeclareAcronym{hoa}{
        short = {HoA},
        long = {height of ambiguity},
}

\begin{document}
\title{ Sensing Accuracy Optimization for Communication-assisted Dual-baseline UAV-InSAR \\\vspace{-2mm}
\thanks{  This work was supported in part by the Deutsche Forschungsgemeinschaft (DFG, German Research Foundation) GRK 2680 – Project-ID 437847244.}
}
\author{\IEEEauthorblockN{Mohamed-Amine~Lahmeri\IEEEauthorrefmark{1}, Víctor Mustieles-Pérez\IEEEauthorrefmark{1}\IEEEauthorrefmark{2}, Martin Vossiek\IEEEauthorrefmark{1}, Gerhard Krieger\IEEEauthorrefmark{1}\IEEEauthorrefmark{2}, and
Robert Schober\IEEEauthorrefmark{1}}\\ \vspace{-2mm}
\IEEEauthorblockA{\IEEEauthorrefmark{1}Friedrich-Alexander-Universit\"at Erlangen-N\"urnberg (FAU), Germany\\
\IEEEauthorrefmark{2}German Aerospace Center (DLR),  Microwaves and Radar Institute, Weßling, Germany\\
\vspace{-8mm}}}
\maketitle
\begin{abstract} 
In this paper, we study the optimization of {the} sensing accuracy of unmanned aerial vehicle (UAV)-based dual-baseline interferometric synthetic aperture radar (InSAR) systems. A swarm of three UAV-synthetic aperture radar (SAR) systems is deployed to image an area of interest from  different angles, enabling the creation of two independent digital elevation models (DEMs). To reduce the InSAR sensing error, i.e., the height estimation error,  the two DEMs are fused based on weighted {averaging} techniques into one final DEM. The heavy computations {required for this process} are performed on the ground. {To this end}, the radar data is offloaded in real time via a frequency division multiple access (FDMA) air-to-ground backhaul link. In this work, we focus on {improving the sensing accuracy by minimizing the worst-case height estimation error of the final DEM. {To this end}, the UAV formation and the power allocated for offloading are jointly optimized based on alternating optimization (AO), while meeting practical InSAR sensing and communication constraints.} Our simulation results demonstrate that the proposed solution can {significantly improve the sensing accuracy compared to classical single-baseline UAV-InSAR systems and other benchmark schemes. }
\end{abstract}
\section{Introduction}
The use of \ac{uav} swarms {for} remote sensing has {recently} gained attention due to their flexibility and efficiency in data collection tasks \cite{data_collection}. This has led to an increased use of drones in diverse applications, such as mapping, monitoring traffic, and addressing climate change \cite{climate_change}. {For} these tasks, a {variety} of sensors can be deployed onboard, {including} cameras, LiDARs, and radars. In particular, the {deployment} of \ac{sar} on \acp{uav} {has} attracted significant {interest} due to  its ability to provide {very} high-resolution {\ac{sar} images over local areas}, even {under} challenging conditions, {overcoming} the limitations of traditional airborne and spaceborne systems. This integration has sparked multiple recent studies focusing on system design \cite{3D_sensing}, trajectory and resource allocation optimization \cite{amine1,amine2,sun2}, and experimental measurement campaigns for \ac{uav}-\ac{sar} systems \cite{experimental}.\par
An interesting remote sensing application {of} \ac{uav} swarms is \ac{3d} radar imaging, which can be {realized} using techniques, such as \ac{mimo} radar, tomography, and interferometry \cite{cramer}. In particular, \ac{insar} leverages the phase differences between at least two \ac{sar} images, captured from different angles, to extract {topographic information and generate \acp{dem}}. Key performance metrics in interferometry include \ac{snr},  coverage,  coherence, \ac{hoa}, and height error \cite{coherence1,snr_equation}, which are affected by the interferometric baseline, i.e., the distance between the sensing platforms. While \ac{insar} has been extensively studied for spaceborne and airborne platforms \cite{cramer}, the optimization of \ac{insar} performance for \ac{uav}-based systems remains largely unexplored. In our recent research work \cite{amine3}, we investigated formation and resource allocation optimization for maximizing the \ac{insar} coverage, but for a single-baseline \ac{uav}-\ac{insar} system. { In contrast, dual-baseline \ac{insar} systems {offer} advantages, such as enhanced phase unwrapping and improved sensing accuracy \cite{fusion_sigma_h}. However, results from  single-baseline systems \cite{amine3} do not apply to {dual-baseline systems} due to the different {expressions} for the height error and the use of multiple  acquisition geometries.} \par
In this work, we study a dual-baseline \ac{uav}-based \ac{insar} system, where a swarm consisting of one master and two slave \acp{uav} is deployed to generate two independent \acp{dem} of a target area, which are {then} fused into a single {\ac{dem}} based on weighted averaging \cite{fusion_sigma_h}. Additionally, the radar data is offloaded to the ground in real time. We  investigate the joint optimization of the \ac{uav} formation and communication power allocation for minimization of the { worst-case} height error {in} the final \ac{dem} under communication and sensing constraints. Our main contributions can be summarized as follows:\begin{itemize}
\item {We propose an approximate bi-static \ac{snr} expression valid for the considered sensing application.}
\item Based on the Cramér–Rao bound of the phase error, we derive a tractable upper bound for the complex expression of the height error of the final \ac{dem}.
\item We formulate and solve a joint optimization problem for \ac{uav} formation and communication power allocation to minimize the derived upper bound on the height error, while satisfying sensing and communication constraints. 
\item Our simulation results demonstrate the effectiveness of the considered dual-baseline \ac{insar} system compared to single-baseline systems and other benchmark schemes.
\end{itemize}
{\em Notations}:
In this paper, lower-case letters $x$ refer to scalar {variables}, while boldface lower-case letters $\mathbf{x}$ denote vectors.  $\{a, ..., b\}$ denotes the set of all integers between $a$ and $b$.  $|\cdot|$ denotes the absolute value operator. $\mathbb{R}^{N}$ represents the set of all $N$-dimensional vectors with real-valued entries. For a vector {$\mathbf{x}=(x_1,...,x_N)^T\in\mathbb{R}^{N}$}, $||\mathbf{x}||_2$ denotes the Euclidean norm, whereas  $\mathbf{x}^T$ stands for the  transpose of $\mathbf{x}$.  For real-valued multivariate functions $f(\mathbf{x})$, $\frac{\partial f}{\partial \mathbf{x}}(\mathbf{a})=\Big(\frac{\partial f}{\partial x_1}(\mathbf{a}),...,\frac{\partial f}{\partial x_N}(\mathbf{a})\Big)^T$ denotes the partial derivative of  $f$ \ac{wrt} $\mathbf{x}$ evaluated for an arbitrary vector $\mathbf{a}$. For any {Boolean} expression $\mathcal{S}$, $\mathds{1}\{\mathcal{S}\}$ denotes the indicator function, which equals 1 if $\mathcal{S}$ is {true and 0 otherwise.}
\section{System Model} \label{Sec:SystemModel}
We consider three rotary-wing \acp{uav}, {denoted} by $U_k, k \in \{0,1,2\}$, performing \ac{insar} sensing over a target area. $U_0$, the master drone, transmits and receives radar echoes, while $U_1$ and $U_2$, the slave drones, only receive. We use a \ac{3d} coordinate system, where the  $x$-,  $y$-, and $z$-{axes} represent the range {direction}, the azimuth {direction}, and the altitude, respectively. The mission time $T$ is divided into $N$ slots {of duration $\delta_t$}, with $T = N\cdot\delta_t$. The drone swarm {forms} a dual-baseline interferometer with two independent observations {acquired by} $(U_0, U_1)$ and $(U_0, U_2)$, {respectively}.  The considered \ac{uav}-\ac{sar} systems operate in stripmap mode \cite{book1} and fly at a constant velocity, $v_y$, following a linear trajectory  that is parallel to a line, denoted by $l_t$, {which is} parallel to {the} $y$-axis and {passes} in time slot $n$ through reference point $\mathbf{p}_t[n]=(x_t,y[n],0)^T \in \mathbb{R}^3$, see Figure \ref{fig:system-model}. The position of $U_k$ {in} time slot $n \in \{1,...,N\}$ is $\mathbf{q}_k[n]=(x_k,y[n],z_k)^T$, with the $y$-axis position vector $\mathbf{y}=(y[1]=0,y[2], ..., y[N])^T\in\mathbb{R}^{N}$ given by: 
 \begin{align}
 y[n+1]=y[n]+v_y\delta_t, \forall n \in \{1,N-1\}.
 \end{align}
For simplicity, we denote the position of $U_k$ in the across-track plane (i.e., $xz-$plane) by $\mathbf{q}_k=(x_k,z_k)^T \in \mathbb{R}^2, \forall k \in \{ 0,1,2\}$. The interferometric baseline, $b_k$, which refers to the distance between sensors $U_0$ and $U_k$, is given by:
\begin{align}
   b_k(\mathbf{q}_0,\mathbf{q}_k) = ||\mathbf{q}_k -\mathbf{q}_0||_2, \forall k \in \{1,2\}.
\end{align}
 The perpendicular baseline, {denoted by $b_{\bot,k}$,} is the magnitude of the projection of $U_k$'s baseline vector perpendicular to $U_0$'s \ac{los} to $\mathbf{p}_t[n]$ and is given by:
\begin{equation}
	b_{\bot,k}(\mathbf{q}_0,\mathbf{q}_k)=
	b_k(\mathbf{q}_0,\mathbf{q}_k)  \cos\Big(\theta_0- \alpha_k(\mathbf{q}_0,\mathbf{q}_k )\Big), \forall k \in \{1,2\},
\end{equation} 
where $\theta_0$ is the fixed look angle that $U_0$'s \ac{los} has with the vertical, and $\alpha_k$ is the angle between the interferometric baseline $b_k$ and the horizontal plane. 
\subsection{\ac{insar} Performance}
Next, we introduce the relevant \ac{insar} sensing performance metrics.
\subsubsection{\ac{insar} Coverage}
Let $r_k, k \in \{0,1,2\}$, {denote} $U_k$'s slant range \ac{wrt} $ \mathbf{p}_t[n]$. {The} slant range is independent of time and is given by:
\begin{align}
	r_k(\mathbf{q}_k)= \sqrt{ (x_k -x_t)^2 +  z_k^2  }, \forall k \in \{0,1,2\}.
\end{align}The radar swath is designed to be centered \ac{wrt} $l_t$. To this end, the look angle of the slave \acp{uav}, denoted by $\theta_k(\mathbf{q}_k), k \in \{1,2\}$, is adjusted {such that} the beam footprint is centered around $\mathbf{p}_t$, i.e., $\theta_k(\mathbf{q}_k)= \arctan\left(\frac{x_t-x_k}{z_k}\right)$. The swath width of $U_k$ can be approximated as follows \cite{book1}:
\begin{align}\label{eq:swath_width}
	S_k(\mathbf{q}_k) = \frac{\Theta_{\rm 3 dB}r_k(\mathbf{q}_k) }{\cos(\theta_k(\mathbf{q}_k))},\forall k \in \{0,1,2\},
\end{align}
where   $\Theta_{\mathrm{3dB}}$ is the -3 dB beamwidth in elevation.
\begin{figure}
	\centering
	\ifonecolumn
	\includegraphics[width=4in]{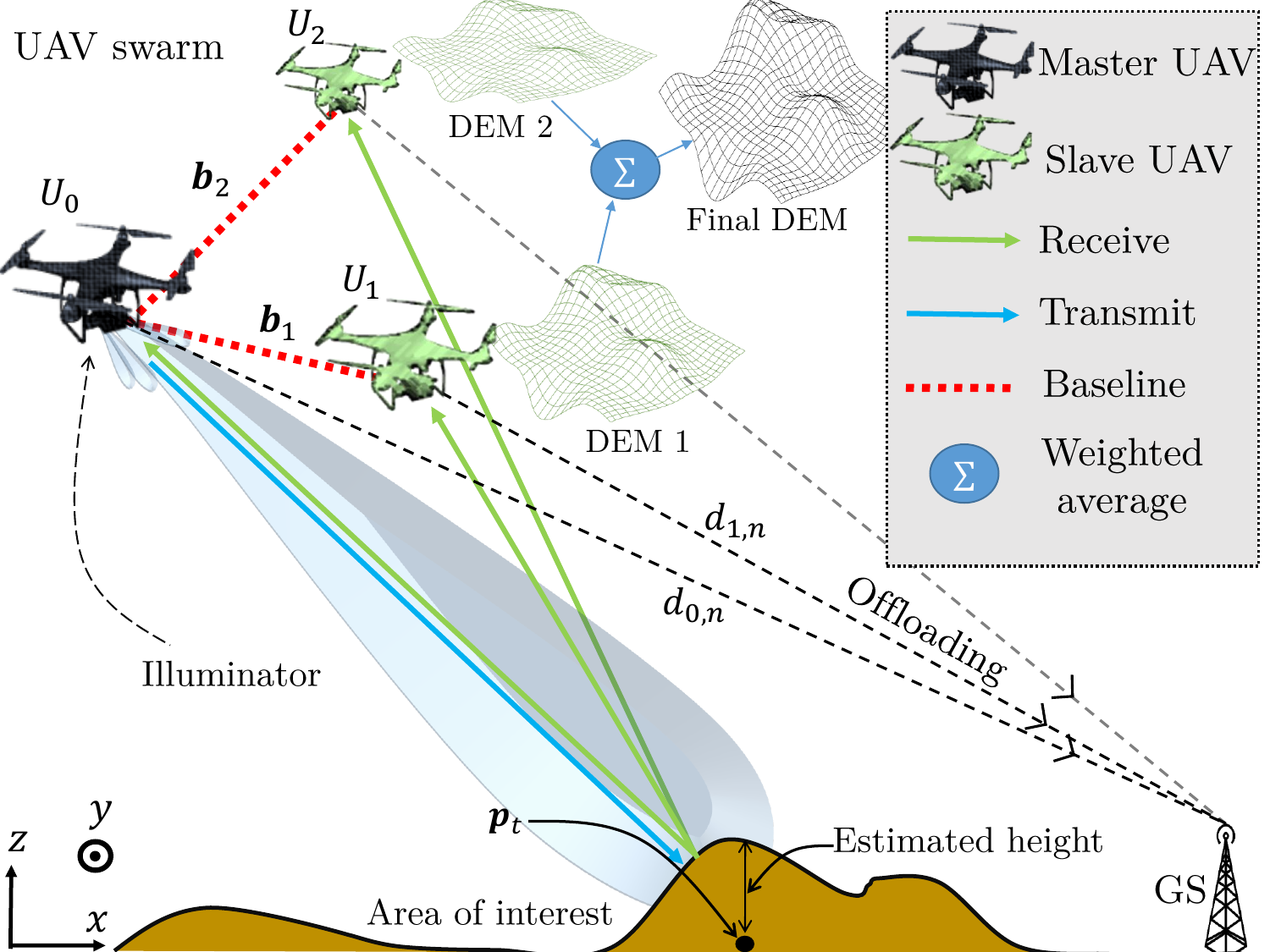}
	\fi
	\caption{ Dual-baseline \ac{insar} sensing system {comprising} one master {and} two slave \ac{uav}-\ac{sar} systems {as well as} a \ac{gs} for real-time data offloading.}
	\label{fig:system-model}
\end{figure} 
\subsubsection{\Ac{insar} Coherence} A key performance metric for \ac{insar} is coherence, representing the cross-correlation between two \ac{sar} images. For {the} {images acquired by $(U_0,U_k), k \in \{1,2\}$,} the total coherence can be decomposed into {several} decorrelation sources as follows: 
\begin{equation}
	\gamma_{k}(\mathbf{q}_0,\mathbf{q}_k)= \gamma_{\mathrm{Rg},k}(\mathbf{q}_k) \gamma_{\mathrm{SNR},k}(\mathbf{q}_0,\mathbf{q}_k)\gamma_{\rm other}, \forall k \in \{1,2\},
\end{equation} where $\gamma_{\mathrm{Rg},k}$ is the baseline  decorrelation, $\gamma_{\mathrm{SNR},k}$ is the \ac{snr} decorrelation, and $\gamma_{\rm other}$ represents the contribution from all other decorrelation sources. The \ac{snr} decorrelation of pair $(U_0,U_k)$ is affected by {the} \acp{snr} {of both \acp{uav}} and is given by  \cite{snr_equation}:
 \begin{align} \label{eq:snr_decorrelation}
	&\gamma_{\mathrm{SNR},k}(\mathbf{q}_0,\mathbf{q}_k)= \frac{1}{\sqrt{1+\mathrm{SNR}^{-1}_{0}(\mathbf{q}_0)}}\frac{1}{\sqrt{1+\mathrm{SNR}^{-1}_{k}(\mathbf{q}_0,\mathbf{q}_k)}}, 
\end{align}
where $\mathrm{SNR}_0$ denotes {the \ac{snr} of the mono-static acquisition by $U_0$}  given by \cite{snr_equation}:
\begin{equation}\label{eq:monostatic_snr}
\mathrm{SNR}_{0}(\mathbf{q}_0)=\frac{\gamma_m}{ r_0^3(\mathbf{q}_0)},
  \end{equation}
where $\gamma_m=\frac{\sigma_0 P_t\; G_t\; G_r \lambda^3 c \tau_p \mathrm{PRF} }{4^4 \pi^3  v_y \sin(\theta_0) k_b T_{\mathrm{sys}}  \; B_{\mathrm{Rg}} \; F \; L}$. Here, $\sigma_0$ is the normalized backscatter coefficient, $P_t$ is the radar transmit power, $G_t$ and $G_r$ are the transmit and receive antenna gains, respectively, $\lambda$ is the radar wavelength, $c$ is the speed of light, $\tau_p$ is the pulse duration, $\mathrm{PRF}$ is the pulse repetition frequency, $k_b$ is the Boltzmann constant, $T_{\mathrm{sys}}$ is the receiver temperature, $B_{\mathrm{Rg}}$ is the bandwidth of the radar pulse, $F$ is the noise figure, and  $L$ represents the total radar losses. {The} derivation of the bi-static \ac{snr} for $U_k$, denoted by $\mathrm{SNR}_k$,  is more complicated. {Here}, {assuming a small bi-static angle $|\theta_0-\theta_k|$, which holds for \ac{insar} applications \cite{cramer},} we propose the following approximation\footnote{
	\ifarxiv Please find {a} detailed {derivation of} the {approximated} bi-static \ac{snr} expression in Appendix \ref{app:bistatic}.
	\fi}:
\begin{equation}\label{eq:bistatic_snr}
\mathrm{SNR}_{k}(\mathbf{q}_0,\mathbf{q}_k)\approx \frac{\gamma_{m}}{ r_0^2(\mathbf{q}_0)r_k(\mathbf{q}_k)}, \forall k \in \{1,2\}.
\end{equation}
Furthermore, the baseline decorrelation reflects the loss of coherence caused by the different angles used for the acquisition of {both {\ac{insar}} images} \cite{victor2}: 
\begin{equation} \gamma_{\mathrm{Rg},k}(\mathbf{q}_k)=\frac{1}{B_p} \left[ \frac{2+B_p}{1+\mathcal{X}(\mathbf{q}_k)} -\frac{2-B_p}{1+\mathcal{X}^{-1}(\mathbf{q}_k)} \right],\label{eq:baseline_decorrelation}
\end{equation}
where $B_{ p}=\frac{B_{\mathrm{Rg}}}{f_0}$ is the fractional bandwidth, $f_0$ is the radar center frequency, and function $\mathcal{X}$ is {given} by \cite{victor2}:
\ifonecolumn 
\begin{equation}
\mathcal{X}(\mathbf{q}_k)=\frac{2 \Big( \sin(\theta_0) \mathds{1}\{\theta_0>\theta_k(\mathbf{q}_k)\} + \sin(\theta_k(\mathbf{q}_k))\mathds{1}\{\theta_0\leq\theta_k(\mathbf{q}_k)\}\Big)}{\sin(\theta_0)+\sin(\theta_k(\mathbf{q}_k))}.\label{eq:baseline_decorrelation_X}
\end{equation} 
\fi

\subsubsection{Height of Ambiguity (HoA)} The \ac{hoa} is related to the sensitivity of the radar system to topographic height variations \cite{snr_equation}. {The \ac{hoa} of pair $(U_0,U_k)$  is given by \cite{snr_equation}}: 
\begin{align}\label{eq:height_of_ambiguity}
    h_{\mathrm{amb},k}(\mathbf{q}_0,\mathbf{q}_k)=\frac{\lambda r_0(\mathbf{q}_0) \sin(\theta_0)}{ b_{\perp,k}(\mathbf{q}_0,\mathbf{q}_k)}, \forall k \in \{1,2\}.
\end{align}
\subsubsection{{\ac{dem}} Height Accuracy}
The height error of {the \ac{dem} {acquired by} the \ac{insar}} pair $(U_0,U_k)$ is given by \cite{snr_equation}:
\begin{equation} \label{eq:height_error}
	\sigma_{h_k}(\mathbf{q}_0,\mathbf{q}_k) =h_{\mathrm{amb},k}(\mathbf{q}_0,\mathbf{q}_k) \frac{\sigma_{\Phi_k}(\mathbf{q}_0,\mathbf{q}_k) }{2\pi}, \forall k \in \{1,2\},
\end{equation}
 where {$\sigma_{\Phi_k}$} is the random error in the interferometric phase and can be approximated in the case of high interferometric coherences by the Cramér–Rao bound  \cite{cramer}: 
\begin{equation}\label{eq:phase_error}
	\sigma_{\Phi_k}(\mathbf{q}_0,\mathbf{q}_k)=\frac{1}{\gamma_{k}(\mathbf{q}_0,\mathbf{q}_k)} \sqrt{ \frac{1-\gamma_{k}^2(\mathbf{q}_0,\mathbf{q}_k)}{2n_L}}, \forall k \in \{1,2\},
\end{equation}
{where $n_L$ is the number of independent looks {employed, i.e.,} {$n_L$} adjacent pixels of the interferogram {are averaged} to improve phase estimation \cite{cramer}.} The fusion of the two independent \ac{insar} {\acp{dem}} is performed based on inverse-variance weighting, such that the {height of} an arbitrary target {estimated} by \ac{insar} pair $(U_0,U_k)$, denoted by $h_k$, {is} weighted by $w_k(\mathbf{q}_0,\mathbf{q}_k)=\frac{1}{	\sigma^2_{h_k}(\mathbf{q}_0,\mathbf{q}_k)}$, $k \in\{1,2\}$, and {averaged} as $\frac{h_1 w_1+h_2 w_2}{w_1+w_2}$. The final height error of the fused \ac{dem} is {characterized} by \cite{fusion_sigma_h}: 
\begin{align}
	\sigma_{h}(\mathbf{q}_0,\mathbf{q}_1,\mathbf{q}_2)= \sqrt{\frac{\sum\limits_{k \in \{1,2\}}w_k^2(\mathbf{q}_0,\mathbf{q}_k)\sigma^2_{h_k}(\mathbf{q}_0,\mathbf{q}_k)}{\left(\sum\limits_{k \in \{1,2\}}w_k(\mathbf{q}_0,\mathbf{q}_k)\right)^2}}. \label{eq:final_height_error}
\end{align}
\subsection{Communication Performance}
We consider real-time offloading of the radar data to a \ac{gs}, where the master and slave \acp{uav} employ \ac{fdma}. The  instantaneous   communication transmit power  consumed by \ac{uav} $U_k$ is given by $\mathbf{P}_{\mathrm{com},k}=(P_{\mathrm{com},k}[1],...,P_{\mathrm{com},k}[N])^T \in \mathbb{R}^N, k \in \{0,1,2\}$.
We denote the location of the \ac{gs} by $\mathbf{g}= (g_x, g_y, g_z)^T \in \mathbb{R}^3$ and the distance from $U_k$ to the \ac{gs} by    $d_{k,n}(\mathbf{q}_k) = ||\mathbf{q}_k[n]-\mathbf{g} ||_2, \forall k \in \{0,1,2\}, \forall n.$ Thus, {adopting} the free-space path loss model and \ac{fdma}, the
instantaneous throughput  from $U_k, \forall k\in \{0,1,2\},$ to the \ac{gs} is given by:
\begin{align}
 &R_{k,n}(\mathbf{q}_k,\mathbf{P}_{\mathrm{com},k})= B_{c,k} \; \log_2\left(1+\frac{P_{\mathrm{com},k}[n] \;\beta_{c,k}}{d_{k,n}^2(\mathbf{q}_k)}\right), \forall n,
\end{align}
where  {$B_{c,k}$ is $U_k$'s fixed} communication bandwidth and $\beta_{c,k}$ is {the} reference channel gain\footnote{The reference channel gain is the channel power gain at a reference distance of 1 m.} divided by the noise variance. 
\section{Problem Formulation}
In this paper, we aim to minimize the height error of the final \ac{dem} $\sigma_h$ by jointly optimizing the \ac{uav} formation $\mathcal{Q}=\{\mathbf{q}_k,\forall k \in \{0,1,2\}\}$ and the instantaneous communication transmit powers  $\mathcal{P}=\{\mathbf{P}_{\mathrm{com},k}, \forall k \in \{0,1,2\}\}$, while satisfying communication and sensing quality-of-service constraints. To this end, we formulate  the following optimization problem: 
\begin{alignat*}{2} 
&(\mathrm{P}):\min_{\mathcal{Q},\mathcal{P}} \hspace{3mm}  \sigma_{h}(\mathcal{Q})   & \qquad&  \\
\text{s.t.} \hspace{3mm} &  \mathrm{C1: } \; z_{\mathrm{min}} \leq z_k \leq z_{\mathrm{ max}}, \forall k  \in \{0,1,2\},               &      &  \\ & \mathrm{C2}: \;  x_0=x_t - z_0 \tan(\theta_0),             &      &  \\& \mathrm{C3}: \;  \theta_{\mathrm{min}} \leq \theta_k(\mathbf{q}_k) \leq \theta_{\mathrm{max}}, \forall k \in \{1,2\},            &      &  \\&  \mathrm{C4}:    ||\mathbf{q}_i-\mathbf{q}_j||_2 \geq d_{\mathrm{min}},\forall i \neq j  \in \{0,1,2\},  &      &     
 \\
 &  \mathrm{C5}:  S_k(\mathbf{q}_k)\geq S_{\rm min}, \forall k \in \{0,1,2\},        &      &     
 \\
  &  \mathrm{C6}:  \gamma_{\mathrm{SNR},k}(\mathbf{q}_0,\mathbf{q}_k)\geq \gamma_{\rm SNR}^{\mathrm{min}}, \forall k \in \{1,2\},        &      &     
 \\  &  \mathrm{C7}:  \gamma_{\mathrm{Rg},k}(\mathbf{q}_k)\geq \gamma_{\rm Rg}^{\mathrm{min}}, \forall k \in \{1,2\},             &      &     
 \\
 &    \mathrm{C8}: \;  h_{\mathrm{amb},k}(\mathbf{q}_0,\mathbf{q}_k)\geq  h_{\mathrm{amb}}^{\mathrm{min}}, \forall k \in \{1,2\},             &      & \\
& \mathrm{C9}: 0 \leq P_{\mathrm{com},k}[n]  \leq P_{\mathrm{com}}^{\mathrm{max}}, \forall \; k \in \{0,1,2\}, \forall n,    & &\\
& \mathrm{C10}: R_{k,n}(\mathbf{q}_k,\mathbf{P}_{\mathrm{com},k}) \geq R_{\mathrm{min},k}, \forall \; k \in \{0,1,2\}, \forall n,       & &  \\
& \mathrm{C11}: 	\sum_{n=1}^{N} P_{\mathrm{com},k}[n] \leq E_{\rm com}^{\rm max}, \forall \; k \in \{0,1,2\}.   & &  
\end{alignat*}
Constraint $\mathrm{ C1}$ specifies the maximum and minimum allowed flying altitude, denoted by $z_{\mathrm{max}}$ and $z_{\mathrm{min}}$,  respectively. Constraint $\mathrm{ C2}$ ensures that the beam footprint of the master {\ac{uav}} is centered around $\mathbf{p}_t[n]$. Constraint $\mathrm{ C3}$ specifies the minimum and maximum slave look angle, denoted by $\theta_{\rm min}$ and $\theta_{\rm max}$, respectively.  Constraint $\mathrm{ C4}$ ensures {a minimum safety distance $d_{\mathrm{min}}$  between {any two} \acp{uav}}. Constraint $\mathrm{C5}$ imposes a minimum radar swath width $S_{\mathrm{min}}$. Constraints $\mathrm{C6}$ and $\mathrm{C7}$  ensure minimum thresholds for \ac{snr} and baseline decorrelation,  $\gamma_{\rm SNR}^{\mathrm{min}}$ and $\gamma_{\rm Rg}^{\mathrm{min}}$, respectively. Constraint $\mathrm{C8}$ imposes a minimum \ac{hoa}, $h_{\mathrm{amb}}^{\mathrm{min}}$, {required for} phase unwrapping \cite{coherence1}.  {Constraint $\mathrm{ C9}$ imposes a maximum {communication transmit power},  $P_{\mathrm{com}}^{\mathrm{max}}$.}  Constraint $\mathrm{ C10}$ ensures {the} minimum {required} sensing data rate for $U_k$, $R_{\mathrm{min},k}, \forall k \in \{0,1,2\}$. Constraint $\mathrm{ C11}$ limits the  total communication energy of  $U_k$ to  $E_{\mathrm{com}}^{\rm max}, \forall k \in \{0,1,2\}$.\par 
Problem $\mathrm{(P)}$ is a non-convex optimization problem. {The non-convexity is caused by the} objective function and constraints $\mathrm{C4},\mathrm{C5}$, $\mathrm{C7}$, and $\mathrm{C8}$. In fact, the height error {is simultaneously \ac{hoa}- and coherence-dependent, making} the objective function {challenging}. Moreover, {the lower bound on an Euclidean distance in  $\mathrm{C4}$ and {the} trigonometric functions in  $\mathrm{C5}$,  $\mathrm{C7}$, and $\mathrm{C8}$ make {these constraints} non-convex and difficult to {handle}.} 
\section{Solution of the Optimization Problem}
\begin{figure}
	\centering
	\ifonecolumn
	\includegraphics[width=4.5in]{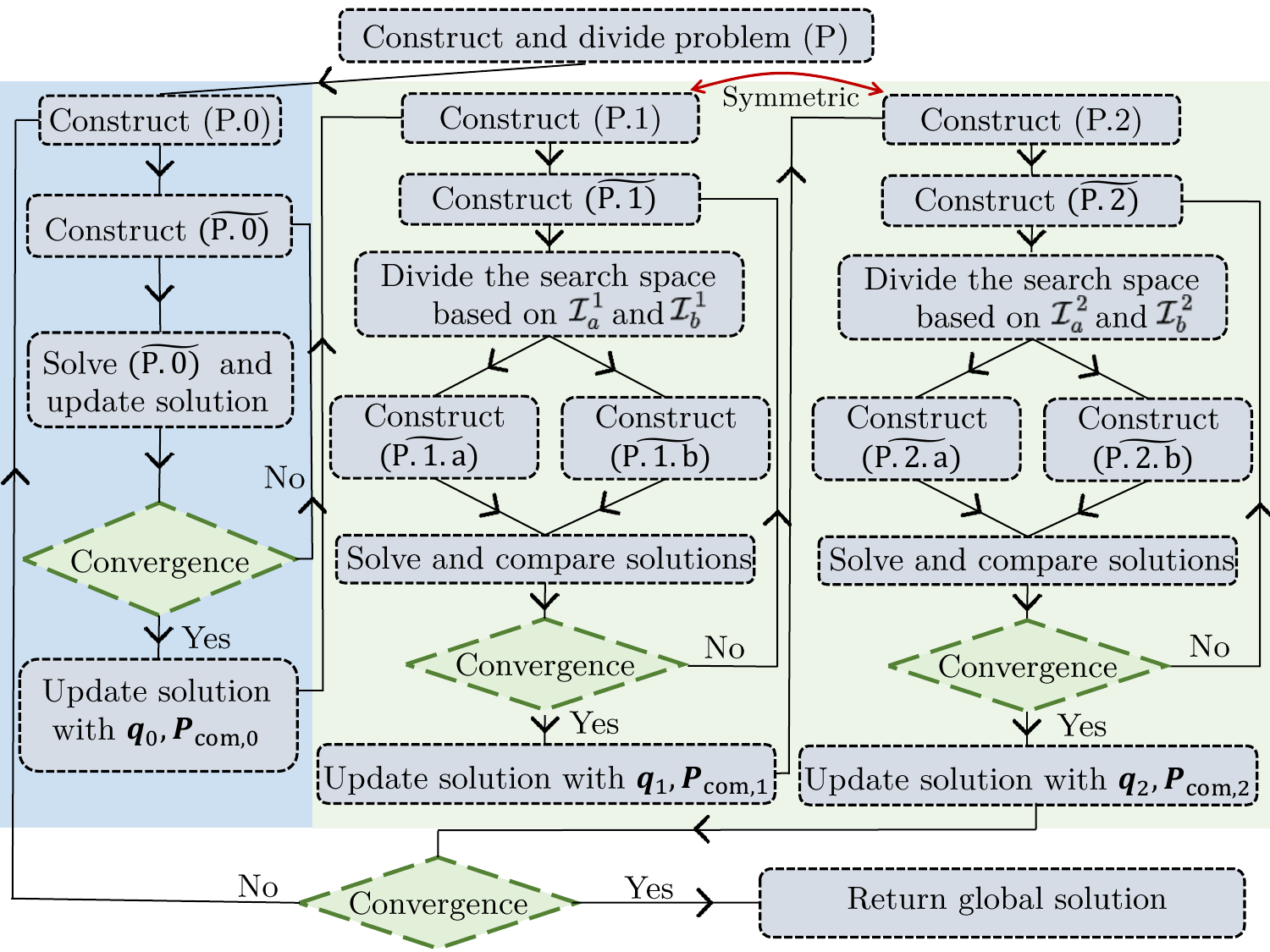}
	\fi
	\caption{Block diagram of the proposed  solution to problem $\mathrm{(P)}$ given by \ac{ao}-based \textbf{Algorithm} \ref{alg:ao}.}
	\label{fig:diagram}
\end{figure}
To balance performance and complexity,  we propose a low-complexity sub-optimal solution {that minimizes an upper bound on problem $\mathrm{(P)}$ based on \ac{ao}}. We divide problem $\mathrm{(P)}$ into 3 sub-problems: $\mathrm{(P.0)}$, $\mathrm{(P.1)}$, and $\mathrm{(P.2)}$. In $\mathrm{(P.0)}$, we optimize the  position and communication power of $U_0$, whereas in $\mathrm{(P.1)}$ and $\mathrm{(P.2)}$, we optimize the {positions and communication powers of} $U_1$ and $U_2$, respectively. Due to symmetry, we focus on $\mathrm{(P.0)}$ and $\mathrm{(P.1)}$, as  $\mathrm{(P.2)}$ can be solved {similarly} to $\mathrm{(P.1)}$, see Figure \ref{fig:diagram}. 

\subsection{Master \ac{uav} Optimization}
In this sub-section, we optimize the position and communication transmit power of the master {\ac{uav}}, denoted by  $\mathbf{q}_0$ and $\mathbf{P}_{\rm com,0}$, respectively, for fixed $\{ \mathbf{q}_1, \mathbf{q}_2,\mathbf{P}_{\rm com,1},\mathbf{P}_{\rm com,2}\}$. The resulting sub-problem, denoted by $\mathrm{(P.0)}$, is still  non-convex due to its objective function and $\mathrm{C4}$. Yet, we leverage  \ac{sca} to provide a low-complexity solution {for} $(\mathrm{P.0})$.\newline
As the master look angle $\theta_0$ is fixed, the perpendicular baseline is independent of $\mathbf{q}_0$ and is given by \cite{amine3}: 
\begin{equation}\label{eq:perpendicular_baseline}
	b_{\bot,k}(\mathbf{q}_k)=\frac{\Big|(x_t-x_k)-\tan(\theta_0)z_k\Big|}{\sqrt{\tan(\theta_0)^2+1}}, \forall k \in \{1,2\}. 
\end{equation} 
\begin{proposition}\label{prop:objective_function_upperbound}
The height error of the final \ac{dem}, $\sigma_{h}$, can be upper bounded based on the worst-case coherence as follows: 
\begin{equation}\label{eq:heigh_error_upper_bound}
	\sigma_{h}(\mathcal{Q})\leq  \overline{\sigma_{h}}(\mathcal{Q}) \triangleq \frac{ }{}\sqrt{\frac{\lambda^2 r^2_0(\mathbf{q}_0 )\sin^2(\theta_0) (1-\mathcal{A}^2)}{  8 \pi^2 \mathcal{A}^2 n_L \left( b^2_{\bot,1}(\mathbf{q}_1)+ b^2_{\bot,2}(\mathbf{q}_2) \right)}}, 
\end{equation}
	where $\mathcal{A}= \gamma_{\rm Rg}^{\mathrm{min}} \gamma_{\rm SNR}^{\mathrm{min}} \gamma_{\rm other}$. 
\end{proposition}
\begin{proof}
	\ifarxiv
	Please refer to Appendix \ref{app:objective_function_upperbound}.
	\fi
\end{proof}
{Therefore},  we relax the complex objective function {of $\mathrm{(P.0)}$} by minimizing instead {the} upper bound on the height error, denoted by $\overline{\sigma_{h}}$ and provided {in} \textbf{Proposition} \ref{prop:objective_function_upperbound}. In iteration $i$ of the \ac{sca} algorithm, constraint $\mathrm{C4}$ is tackled using a surrogate function for the Euclidean distance $||\mathbf{q}_0-\mathbf{q}_k||_2$ around $\mathbf{q}_0^{(i)}$ as follows  \cite{surrogate}: 
\begin{equation} \label{eq:C4_constraint}
\mathrm{\widetilde{C4}}: 2 \mathbf{q}_0^T (\mathbf{q}_0^{(i)} -\mathbf{q}_k)-|| \mathbf{q}_0^{(i)}||_2^2+|| \mathbf{q}_k||_2^2\geq d_{\rm min}, {\forall i, } \forall k \in \{1,2\}.
\end{equation}
The resulting sub-problem is denoted by $\widetilde{\mathrm{(P.0)}}$ and is given by: 
\begin{alignat*}{2} 
	&\widetilde{\mathrm{(P.0)}}:\min_{\mathbf{q}_0,\mathbf{P}_{\rm com,0}} \hspace{3mm}  \overline{\sigma_{h}}(\mathcal{Q})   & \qquad&  \\
	\text{s.t.} \hspace{3mm} &\mathrm{C1-C3,\widetilde{\mathrm{C4}}, C5,C6, C8-C11}. &       &     
\end{alignat*}
Problem $\widetilde{\mathrm{(P.0)}}$ is convex and can be solved using the Python CVXPY library \cite{cvxpy}. The solution procedure to solve $\mathrm{(P.0)}$ is summarized in \textbf{Algorithm} \ref{alg:sca1}, which converges to a {local optimum of the upper bound on problem $\mathrm{(P.0)}$,  $\overline{\sigma_h}$,} in polynomial computational time
complexity \cite{sca3.5}.  \textbf{Algorithm} \ref{alg:sca1} involves $N+2$ optimization variables, resulting in a computational complexity of $\mathcal{O}(M_0(N+2)^{3.5})$, where $M_0$ is the number of iterations needed for convergence \cite{sca3.5}.
\ifonecolumn
\begin{algorithm}[]
	\caption{Successive Convex Approximation Algorithm}\label{alg:sca1}
	\begin{algorithmic}[1] 
		\label{algorithm1} \State For fixed $\{\mathbf{q}_1,\mathbf{q}_2, \mathbf{P}_{\rm com,1},\mathbf{P}_{\rm com,2}\}$, set initial point $\mathbf{q}_0^{(1)}$, iteration index $i=1$, and error tolerance $0< {\epsilon_0} \ll 1$.
		\State \textbf{repeat}  
		\State\hspace{\algorithmicindent}$ $Determine sensing {worst-case} {accuracy} $\overline{\sigma_{h}}(\mathbf{q}_0,\mathbf{q}_1,\mathbf{q}_2) , \mathbf{q}_0,$ and $ \mathbf{P}_{\mathrm{ com,0}} $ by solving $\mathrm{ \widetilde{(P.0)}}$ around point $\mathbf{q}_0^{(i)}$ using \Statex\hspace{\algorithmicindent}CVXPY       
		\State\hspace{\algorithmicindent}$ $Set $ i=i+1$ and $\mathbf{q}_0^{(i)}=\mathbf{q}_0$ 
		\State \textbf{until} $ \big |\frac{\overline{\sigma_{h}}(\mathbf{q}_0^{(i)},\mathbf{q}_1,\mathbf{q}_2)-\overline{\sigma_{h}}(\mathbf{q}_0^{(i-1)},\mathbf{q}_1,\mathbf{q}_2)}{\overline{\sigma_{h}}(\mathbf{q}_0^{(i)},\mathbf{q}_1,\mathbf{q}_2)}\big|<  {\epsilon_0} $
		\State \textbf{return} solution \{$\mathbf{q}_0, \mathbf{P}_{\mathrm{ com,0}}$\}
	\end{algorithmic}
\end{algorithm}
\fi
\subsection{Slave \ac{uav} Optimization} 
{Next}, we optimize the position and communication transmit power of slave {\ac{uav}} $U_1$, denoted by  $\mathbf{q}_1$ and $\mathbf{P}_{\rm com,1}$, respectively, for fixed $\{ \mathbf{q}_0, \mathbf{q}_2,\mathbf{P}_{\rm com,0},\mathbf{P}_{\rm com,2}\}$.  The resulting problem, denoted by $\mathrm{(P.1)}$, 
is non-convex due {to} {the} objective function and constraints $\mathrm{C4}, \mathrm{C5}, \mathrm{C7},$ and $\mathrm{C8}$. To tackle this sub-problem, we employ {again} \ac{sca}. First, we {adopt} the upper bound $\overline{\sigma_h}$ provided by \textbf{Proposition} \ref{prop:objective_function_upperbound}. Furthermore, it can be shown that minimizing $\overline{\sigma_{h}}$ for fixed $\mathbf{q}_0$ and $\mathbf{q}_2$ is equivalent to maximizing the perpendicular baseline $b_{\bot,1}(\mathbf{q}_1)$. Moreover, {in each \ac{sca} iteration $j$,} non-convex  constraint $\mathrm{C4}$ is replaced with convex constraint $ \widetilde{\widetilde{\mathrm{C4}}}$ based on a surrogate function similar to (\ref{eq:C4_constraint}). Constraint $\mathrm{C5}$ is convexified based on {a} Taylor approximation around point $\mathbf{q}_1^{(j)}$ as follows: 
\begin{equation} 
\mathrm{\widetilde{C5}}: r_1^2(\mathbf{q}_1^{(j)})+ \left(\frac{\partial r_1^2}{\partial \mathbf{q}_1 } (\mathbf{q}_1^{(j)}) \right)^T(\mathbf{q}_1-\mathbf{q}_1^{(j)}) \geq \frac{S_{\rm min} z_1 }{ \Theta_{\rm 3dB}}, {\forall j}. 
\end{equation}
The resulting problem is given by: 
\begin{alignat*}{2} 
	&\mathrm{\widetilde{(P.1)}}:\max_{\mathbf{q}_1,\mathbf{P}_{\rm com,1}} \hspace{3mm}  	b_{\bot,1}(\mathbf{q}_1) & \qquad&  \\
	\text{s.t.} \hspace{3mm} &\mathrm{C1,C3,\widetilde{\widetilde{C4}}, \widetilde{C5}, {C6}-C11}. &       &     
\end{alignat*}
Yet, the {expressions} {for} the perpendicular baseline  in (\ref{eq:perpendicular_baseline}) and  {for} the baseline {decorrelation} still present obstacles {for} solving $\mathrm{\widetilde{(P.1)}}$. Thus, we divide the search space of problem $\mathrm{\widetilde{(P.1)}}$ into two disjoint sets, denoted by $\mathcal{I}_a^k=\{\mathbf{q}_k; \theta_0 \geq \theta_k (\mathbf{q}_k) \}$ and  $\mathcal{I}^k_b=\{\mathbf{q}_k; \theta_0 < \theta_k (\mathbf{q}_k) \}, k \in \{1,2\}$. The solution that maximizes $\overline{\sigma_h}$ over $\mathcal{I}^1_a$ and $\mathcal{I}^1_b$ is selected, see Figure \ref{fig:diagram}. 
\begin{proposition}\label{prop:baseline_decorrelation}
	Constraint $\mathrm{C7}$ is equivalent to the following convex constraints:
	\begin{align}
		\quad 
		\begin{cases}
			\mathrm{C7a}:	z_k \alpha_a -(x_t-x_k) \leq 0, \text{ if } \mathbf{q}_k \in \mathcal{I}_{a}^k\\
			\mathrm{C7b}: (x_t-x_k)-z_k\alpha_b \leq 0,  \text{ if } \mathbf{q}_k \in \mathcal{I}_{b}^k
		\end{cases}, \forall k \in \{1,2\},
	\end{align} 
	where $\alpha_a=\tan\left( \arcsin \left( \frac{2-h(\gamma_{\rm Rg}^{\rm min})}{h(\gamma_{\rm Rg}^{\rm min})}\sin(\theta_0)\right)\right)$, $\alpha_b=\tan\left( \arcsin \left( \frac{ h(\gamma_{\rm Rg}^{\rm min})}{2-h(\gamma_{\rm Rg}^{\rm min})}\sin(\theta_0)\right)\right)${, and function $h(x)=\frac{xB_p  -2-B_p}{B_p-2-xB_p}$}. 
\end{proposition}
\begin{proof}
	\ifarxiv
	Please refer to Appendix \ref{app:baseline_decorrelation}. 
	\fi
\end{proof}
Constraint $\mathrm{C7}$ is transformed based on \textbf{Proposition} \ref{prop:baseline_decorrelation}. Then, for $\mathbf{q}_1 \in \mathcal{I}^1_a$, constraint $\mathrm{C7}$ is replaced by $\mathrm{C7a}$ and the resulting problem is denoted by  $\mathrm{\widetilde{(P.1.a)}}$.  {Simlarly, sub-problem $\mathrm{\widetilde{(P.1.b)}}$ denotes  sub-problem $\mathrm{\widetilde{(P.1)}}$ for $\mathbf{q}_1 \in \mathcal{I}^1_b$, {where} constraint $\mathrm{C7b}$ replaces $\mathrm{C7}$.} \par
The proposed \ac{sca} algorithm to solve  $\mathrm{(P.1)}$ is omitted due to space limitation, but is similar to \textbf{Algorithm} \ref{alg:sca1}, where the convex approximations  $\mathrm{ (\widetilde{P.1.a})}$ and $\mathrm{ (\widetilde{P.1.b})}$ are solved in parallel using CVXPY \cite{cvxpy}, {with precision $\epsilon_1=\epsilon_0$.} The algorithm converges to a local optimum of the upper bound on sub-problem $\mathrm{(P.1)}$ {entailing} computational complexity $\mathcal{O}(2 M_1(N+2)^{3.5})$, where $M_1$ {is} the required number of iterations \cite{sca3.5}.
\ifonecolumn
\begin{algorithm}[h]
	\caption{Alternating Optimization Algorithm}\label{alg:ao}
	\begin{algorithmic}[1] 
		\State Set initial formation $\mathcal{Q}^{(1)}=\{\mathbf{q}_0^{(1)},\mathbf{q}_1^{(1)},\mathbf{q}_2^{(1)}\}$, initial  communication transmit powers $\mathcal{P}^{(1)}=\{\mathbf{P}^{(1)}_{\mathrm{com,0}},\mathbf{P}^{(1)}_{\mathrm{com,1}},\mathbf{P}^{(1)}_{\mathrm{com,2}}\}$, iteration index $m=1$, and error tolerance $0< {\epsilon_2}  \ll 1$.
		\State \textbf{repeat}
		\State\hspace{\algorithmicindent}$ $ Set $ m=m+1$
		\State\hspace{\algorithmicindent}$ $ Determine $ \overline{\sigma_h}(\mathbf{q}_0,\mathbf{q}_1^{(m-1)},\mathbf{q}_2^{(m-1)})$ and set $\mathbf{q}_0^{(m)}=\mathbf{q}_0$ and  $\mathbf{P}^{(m)}_{\mathrm{com,0}}=\mathbf{P}_{\mathrm{com,0}}$ {after} solving $\rm (P.0)$ for  \Statex\hspace{\algorithmicindent} fixed $\{\mathbf{q}_1^{(m-1)},\mathbf{q}_2^{(m-1)},\mathbf{P}^{(m-1)}_{\mathrm{com,1}},\mathbf{P}^{(m-1)}_{\mathrm{com,2}}\}$  using \textbf{Algorithm} \ref{alg:sca1}   
				\State\hspace{\algorithmicindent}$ $ Determine $ \overline{\sigma_h}(\mathbf{q}_0^{(m-1)},\mathbf{q}_1,\mathbf{q}_2^{(m-1)})$ and set $\mathbf{q}_1^{(m)}=\mathbf{q}_1$ and  $\mathbf{P}^{(m)}_{\mathrm{com,1}}=\mathbf{P}_{\mathrm{com,1}}$  {after} solving $\rm (P.1)$\Statex\hspace{\algorithmicindent} for fixed  $\{\mathbf{q}_0^{(m-1)},\mathbf{q}_2^{(m-1)},\mathbf{P}^{(m-1)}_{\mathrm{com,0}},\mathbf{P}^{(m-1)}_{\mathrm{com,2}}\}$  using \ac{sca}
	\State\hspace{\algorithmicindent}$ $ Determine $ \overline{\sigma_h}(\mathbf{q}_0^{(m-1)},\mathbf{q}_1^{(m-1)},\mathbf{q}_2)$ and set $\mathbf{q}_2^{(m)}=\mathbf{q}_2$ and  $\mathbf{P}^{(m)}_{\mathrm{com,2}}=\mathbf{P}_{\mathrm{com,2}}$  {after} solving $\rm (P.2)$ \Statex\hspace{\algorithmicindent} for fixed $\{\mathbf{q}_0^{(m-1)},\mathbf{q}_1^{(m-1)},\mathbf{P}^{(m-1)}_{\mathrm{com,0}},\mathbf{P}^{(m-1)}_{\mathrm{com,1}}\}$  using \ac{sca}
		\State \textbf{until} $\big |\frac{ \overline{\sigma_h}(\mathbf{q}_0^{(m)},\mathbf{q}_1^{(m)},\mathbf{q}_2^{(m)})-\overline{\sigma_h}(\mathbf{q}_0^{(m-1)},\mathbf{q}_1^{(m-1)},\mathbf{q}_2^{(m-1)})}{\overline{\sigma_h}(\mathbf{q}_0^{(m)}, \mathbf{q}_1^{(m)},\mathbf{q}_2^{(m)})}\big|\leq {\epsilon_2} $
		\State \textbf{return} solution $\{\mathcal{Q},\mathcal{P}\}= \{\mathbf{q}_0^{(m)},\mathbf{q}_1^{(m)}, \mathbf{q}_2^{(m)},\mathbf{P}^{(m)}_{\mathrm{com,0}},\mathbf{P}^{(m)}_{\mathrm{com,1}},\mathbf{P}^{(m)}_{\mathrm{com},2}\}$
	\end{algorithmic}
\end{algorithm} 
\fi
\subsection{Solution to Problem $\mathrm{(P)}$} 
To solve problem $\mathrm{(P)}$, we use \ac{ao} by solving sub-problems $\mathrm{(P.0)}$, $\mathrm{(P.1)}$, and $\mathrm{(P.2)}$ iteratively, see Figure \ref{fig:diagram}. In \textbf{Algorithm} \ref{alg:ao}, we summarize all steps used to solve problem  $\mathrm{(P)}$. \textbf{Algorithm} \ref{alg:ao} converges to a local optimum {of the worst-case height error, $\overline{\sigma_h}$,} with time complexity $\mathcal{O}(M_2(2M_1+M_0)(N+2)^{3.5})$, where $M_2$ is the required number of iterations \cite{sca3.5}.\section{Simulation Results}
\begin{table}[]
	\centering
	\caption{System parameters \cite{victor2,snr_equation,coherence1}.}
	\label{tab:my-table}
	\begin{adjustbox}{max width=\columnwidth}
 		\begin{tabular}{|c|c?c|c?c|c|}
	\hline
	Parameter           & Value 					& Parameter & Value 						& Parameter &Value \\ \hline
	$N$               &$80$ 					&$E_{\mathrm{com}}^{\rm max}$&594 J								&$\lambda$&  0.12 m \\ \hline	
	$\delta_t$ 		  &1 s    				&$R_{\mathrm{ min},0}$   &10 Mbit/s   				&$\tau_p\times\mathrm{PRF}$    & 10$^{-4}$  			 \\ \hline
	$z_{\mathrm{min}}$&  1 m                   			&$R_{\mathrm{ min},1}$   &16.95 Mbit/s				&$\theta_d$        & 45°    \\ \hline
	$z_{\mathrm{max}}$&  100 m   					&$R_{\mathrm{ min},2}$   &1 Mbit/s	        		&$T_{\mathrm{sys}}$ & 400 K   \\ \hline
	$h_{\mathrm{amb}}^{\mathrm{min}}$&1.2 m     			&$B_{c,k}, \forall k$   & 1 GHz					&$B_{\mathrm{Rg}}$ &3 GHz\\ \hline
	$x_t$             &20 m       					&$\gamma$          &18.69 dB 	    				&$F$    & 5 dB	  \\ \hline
	$g_x $        &70 m   						&$\gamma_{\mathrm{Rg}}^{\mathrm{min}}$  &0.8		   		&$L$    & 6 dB  \\ \hline
	$g_y$   	&149.37 m      	 				&$\gamma_{\mathrm{SNR}}^{\mathrm{min}}$&$0.8$ 	 			&$f_0$    & 2.5 GHz \\ \hline
	$g_z$  	& 25 m							&$ \gamma_{\rm other} $&    0.6 								&$\epsilon_0=\epsilon_1=\epsilon_2$  &10$^{-2}$ 		\\ \hline
	$d_{\mathrm{min}}$&1.5 m 				 	&$\sigma_0$    &-10 dB							&$\theta_{3\mathrm{dB}}$  &33.44°\\ \hline
	$S_{\rm  min}$   & 55 m			 			&$P_t$			&$26.02$ dBm 		   		& $\theta_{\rm min}$   & 37.24° 	\\ \hline
	$v_y$&  4.3 m/s 						& $G_t$   & 5 dBi				      					&$\theta_{\rm max}  $		&48.7° $  $ \\ \hline
	$P_{\mathrm{com}}^{\mathrm{max}}$&10.1 dB 			& $G_r$   & 5 dBi									&$n_L$&4 				   \\ \hline
\end{tabular}
	\end{adjustbox} \vspace{-4mm}
\end{table}
This section presents simulation results for \textbf{Algorithm} \ref{alg:ao}, using  parameters from Table \ref{tab:my-table}, unless stated otherwise. To evaluate performance, we {adopt} the next benchmark schemes: 
\begin{itemize}
\item{\textbf{Benchmark scheme 1:}} Here, a single-baseline \ac{uav}-\ac{insar} system  consisting only of $U_0$ and $U_1$ is considered \cite{amine3}. The upper bound on  the height error $\sigma_{h_1}$ is minimized based on a two-step \ac{ao} algorithm.
\item{\textbf{Benchmark scheme 2:}} In this scheme, we fix the position of the master \ac{uav} at $\mathbf{q}_0 =\mathbf{q}_0^{\rm fixed}$, which is feasible for $\mathrm{(P)}$, and optimize the {remaining} variables.
\item \textbf{Benchmark scheme 3:} Here, we apply a static and constant communication power allocation (i.e., $P_{\mathrm{com},k}[n] = \frac{E_{\rm com}^{\rm max}}{N}, \forall n, \forall k$), and optimize the {remaining} variables.
\end{itemize}

\begin{figure}[]
	\centering
	\ifonecolumn
	\includegraphics[width=3.5in]{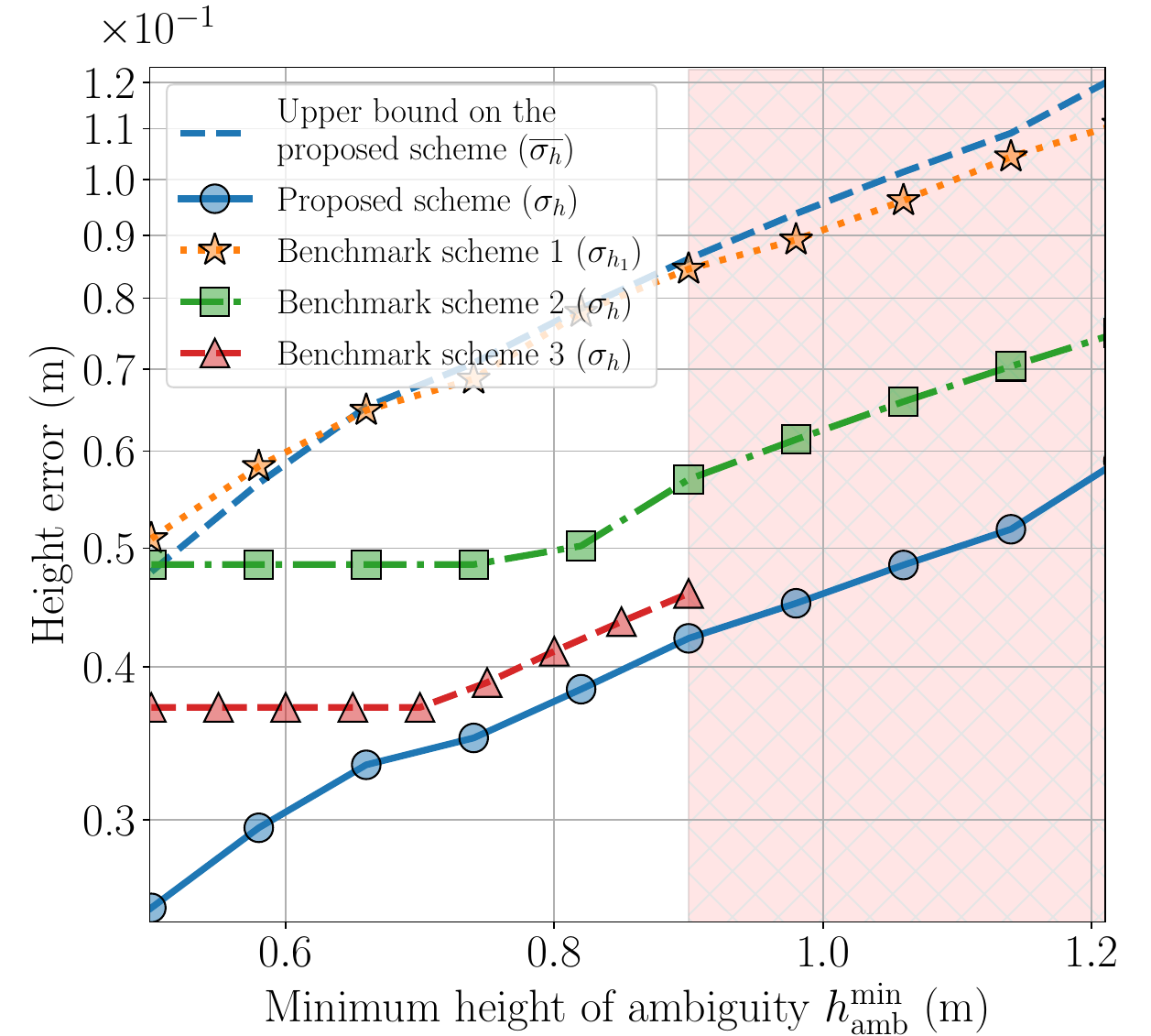}
	\fi
	\caption{Height error of the final \ac{dem}, $\sigma_h$, and its upper bound, $\overline{\sigma_h}$, versus the minimum \ac{hoa}, $h_{\rm amb}^{\rm min}$, for $\mathbf{q}_0^{\rm fixed}=(-54,74)^T$ m {and $P_{\rm com}^{\rm max}=10.1$ dB}. }
	\label{fig:error_vs_minimum_hoa}
\end{figure}
In Figure \ref{fig:error_vs_minimum_hoa}, we present the height error of the final \ac{dem}  $\sigma_h$ and its upper bound $\overline{\sigma_h}$ versus the minimum \ac{hoa}, $h_{\rm amb}^{\rm min}$. The figure shows that the sensing accuracy degrades with stricter requirements on the minimum \ac{hoa}, which is due to the relation between the \ac{hoa} and the height error, see (\ref{eq:height_error}). We note that the tightness of the upper bound on the height error in (\ref{eq:heigh_error_upper_bound}) increases with stricter sensing requirements, i.e., if $\gamma_{\rm SNR}^{\rm min}\gamma_{\rm Rg}^{\rm min} \gamma_{\rm other} \to 1$, then $\sigma_h \to \overline{\sigma_h}$. Nevertheless, Figure \ref{fig:error_vs_minimum_hoa}
{reveals that even for system parameters, for which} the upper bound is not tight, it is a useful metric for optimization. {In fact,} the proposed scheme consistently achieves a gain of at least 49\% compared to benchmark scheme 1. {This gain is {due} to {the} averaging {of} the height error, which improves the sensing accuracy and highlights the importance of using multiple \acp{uav} {for} acquisition.} Additionally, optimizing the \ac{uav} formation {enables} the proposed solution to outperform benchmark scheme 2, with a minimum gain of 23.6\%, {for the system parameters considered in Figure \ref{fig:error_vs_minimum_hoa}}. Compared to benchmark scheme 3, the proposed solution achieves a gain of at least 8.2\%. {Furthermore}, the static power allocation, adopted in benchmark scheme 3, leads to infeasibility of problem $\mathrm{(P)}$ starting from $h_{\rm amb}^{\rm min}=0.9 $ m, indicated by the red colored region in Figure \ref{fig:error_vs_Rmin}.\par
\ifonecolumn \begin{figure}[H]
	\else\begin{figure}[]\fi
	\centering
	\ifonecolumn
	\includegraphics[width=3.5in]{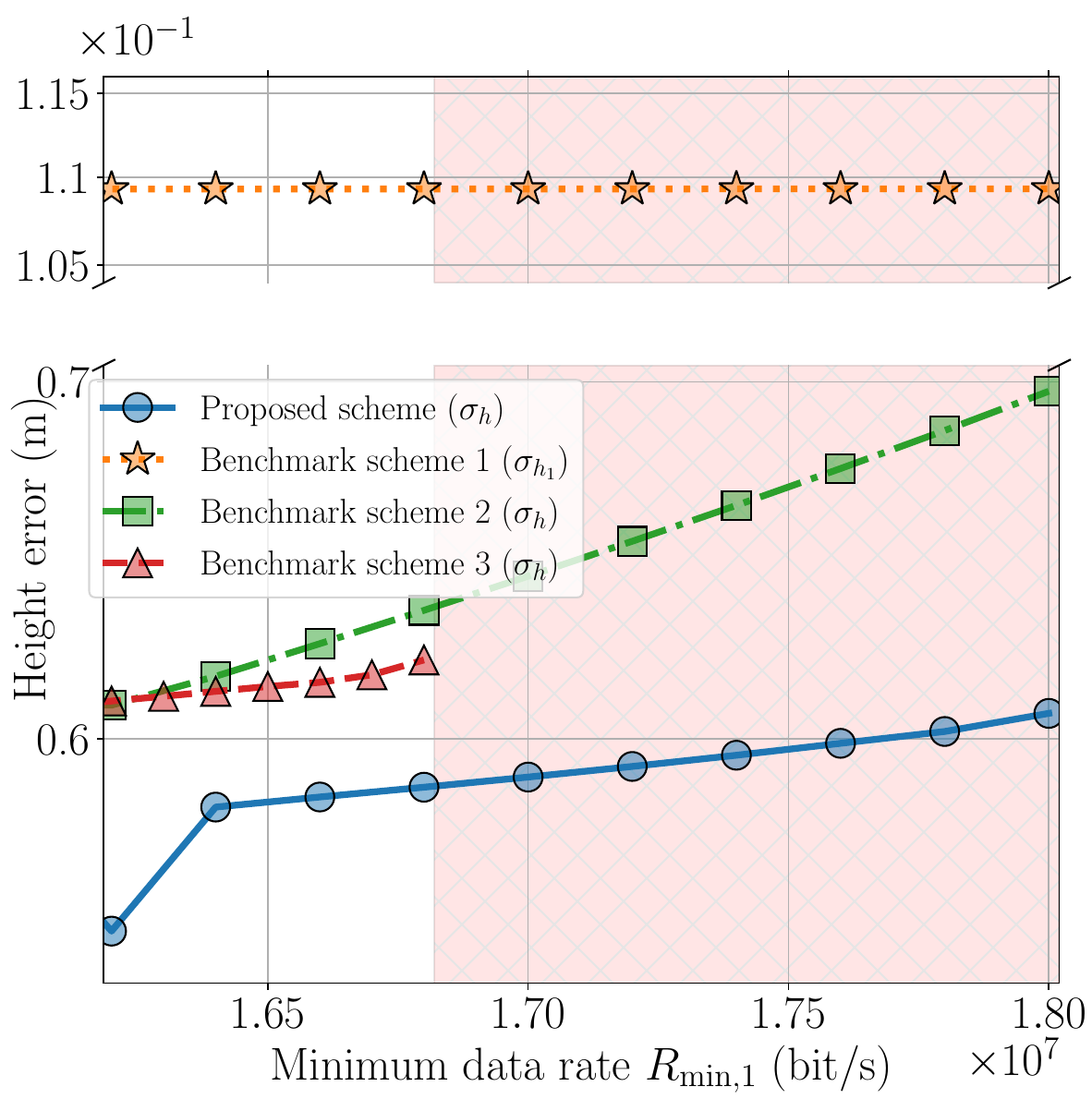}
	\fi
	\caption{Height error of the final \ac{dem}, $\sigma_h$, versus the minimum data rate $R_{\rm min,1}$, for $\mathbf{q}_0^{\rm fixed}=(-41,61)^T$ m and $P_{\rm com}^{\rm max}=9$ dB.}
	\label{fig:error_vs_Rmin}
\end{figure}
{Figure \ref{fig:error_vs_Rmin} depicts the final height error versus the minimum data rate of $U_1$}. The figure shows that, except for benchmark scheme 1,  higher data rate requirements lead to {worse accuracy}. This can be explained by $U_1$'s extended range when lower {data rates} are {sufficient}. In contrast, increasing the data rate {requirement} limits the $U_1$-\ac{gs} distance, which leads {to the activation} of constraint $\mathrm{C10}$. {This in turn}  affects the perpendicular baseline and, therefore, the height error for {the considered} dual-baseline schemes. Furthermore, {the infeasibility of baseline scheme 3 for required rates exceeding $R_{\rm min,1}=16.8$ Mbit/s,} highlights the need to properly allocate the communication power to ensure real-time data offloading to {the} \ac{gs}. Figure \ref{fig:error_vs_Rmin} confirms that the proposed {scheme} {outperforms} all benchmark schemes.
\section{Conclusion}
In this work, we studied a dual-baseline \ac{uav}-based \ac{insar} system using a swarm of three drones to generate two {independent} \acp{dem} of a target area. The final \ac{dem} is obtained {with} a weighted {averaging} technique, improving sensing precision. We proposed a low-complexity algorithm that minimizes {an upper bound on the} height estimation error of the final \ac{dem} by jointly optimizing the \ac{uav} formation and communication power allocation, while meeting sensing and communication constraints. Simulation results showed that the proposed scheme significantly improves sensing accuracy compared to single-baseline systems and other benchmark schemes.

\appendices
\ifarxiv
\section{Bi-static \ac{snr} Approximation} \label{app:bistatic}
In this appendix, we provide the detailed steps for deriving the bi-static \ac{snr} approximation in (\ref{eq:bistatic_snr}). To this end, we start with the bi-static radar \ac{snr} expression for distributed targets, denoted by $\mathrm{SNR}_r$ and given by \cite{book1}: 
\begin{equation}
	\mathrm{SNR}_r=\frac{\sigma_0 A_{\rm res} P_t G_t G_r \lambda^2}{(4\pi)^3 R^2_{\rm Tx}R^2_{\rm Rx}k_b B_{\rm noise} T_{\rm sys}  F L},
\end{equation}
where $R_{\rm Tx}$ and $R_{\rm Rx}$ represent the radar transmit and receive slant ranges, respectively, $A_{\rm res}$ is the size of one resolution cell, {i.e., the size of the smallest area the \ac{sar} system can distinguish}, $B_{\rm noise}$ is the effective noise bandwidth, and $L$ {denotes} the total radar losses. After coherent integration of the radar signal to form the \ac{sar} image, the radar \ac{snr} is improved by a factor $N_{\rm proc}$ as follows:
\begin{equation}
	\mathrm{SNR}=	\mathrm{SNR}_rN_{\rm proc}=\mathrm{SNR}_rN_{\rm pulse} N_a,
\end{equation}  
where $N_{\rm pulse}$ is the number of samples per pulse and $N_a$ is the number of samples within the synthetic aperture length, which is denoted by $L_a$. Note that $N_{\rm pulse}=\tau_p B_{\rm Rg} \approx \tau_p B_{\rm noise}$ and $N_a=\frac{L_a \mathrm{PRF}}{v_y}$. Thus, the image \ac{snr} can be written as: 
\begin{equation} \label{eq:SNR_image}
	\mathrm{SNR}=\frac{\sigma_0 A_{\rm res} P_t G_{t} G_{r} \lambda^2  \tau_p\mathrm{PRF} L_a}{(4\pi)^3 R^2_{\rm Tx}R^2_{\rm Rx} v_y k_b T_{\rm sys} F L}.
\end{equation} 
The  \ac{snr} expression depends on the resolution cell size $A_{\rm res}$, which is related to {the} range and azimuth resolutions, denoted by $\delta_r$ and $\delta_a$, respectively. For a mono-static system with slant range $R_{\rm Tx}$ and incidence angle $\theta$,  $A_{\rm res}=\delta_r \delta_a=\frac{\lambda R_{\rm Tx} c}{4 B_{\rm Rg} L_a \sin(\theta)}$ leads to the expression in (\ref{eq:monostatic_snr}) \cite{snr_equation}. {However,} calculating  $A_{\rm res}$ for the bi-static case is more complex. According to \cite{bistatic_emerging}, {in this case,} the ground range resolution is given by: 
\begin{equation}
	\delta_r=\frac{c}{B_{\rm Rg} \left(\sin(\theta_{\rm Tx})+\sin(\theta_{\rm Rx})\right)},
\end{equation}
where $\theta_{\rm Tx}$ and $\theta_{\rm Rx}$ denote the transmit and receive incidence angles, respectively. Furthermore, the azimuth resolution is given by: 
\begin{equation}
	\delta_a= l \frac{R_{\rm Rx}}{R_{\rm Tx}+R_{\rm Rx}}.
\end{equation}
In the case of across-track interferometry applications, since a lower bound is imposed on the baseline decorrelation by constraint $\mathrm{C7}$, the bi-static angle $\Delta \theta = |\theta_{\rm Tx}-\theta_{\rm Rx}|$ is relatively small (i.e., on the order of few degrees). Thus, the range resolution {can be} approximated by $\delta_r \approx \frac{c}{2 B_{\rm Rg} \sin(\theta_{\rm Tx})}$ and $\sigma_0(\theta_{\rm Rx})\approx \sigma_0(\overbrace{\theta_{\rm Tx}}^{\text{fixed angle}})=\sigma_0$. Based on the bi-static effective range $R_{\rm eff}=\frac{R_{\rm Rx}+R_{\rm Tx}}{2}$, we approximate the azimuth resolution as: 
\begin{equation}
	\delta_a\approx \frac{\lambda R_{\rm eff} R_{\rm Rx}}{L_a (R_{\rm Tx}+R_{\rm Rx})},
\end{equation}
which leads to {the following} approximation {of} the resolution cell size:
\begin{equation} \label{eq:resolution_size_approximation}
		A_{\rm res} =\delta_r \delta_a\approx \frac{\lambda R_{\rm Rx} c }{4 L_a B_{\rm Rg} \sin(\theta_{\rm Tx})}.
\end{equation}
Finally, the approximated bi-static \ac{snr} expression is obtained by {inserting} (\ref{eq:resolution_size_approximation}) {in} (\ref{eq:SNR_image}). 

\else
\fi
\ifarxiv
\section{Proof of Proposition \ref{prop:objective_function_upperbound}} \label{app:objective_function_upperbound}
It can be shown that $\sigma_{\Phi_k}$ is monotonically decreasing \ac{wrt} the total coherence $\gamma_k$. Therefore, given the {worst-case} coherence value $\mathcal{A}= \gamma_{\rm Rg}^{\mathrm{min}} \gamma_{\rm SNR}^{\mathrm{min}} \gamma_{\rm other}$, the following inequality holds: 
\begin{equation}\label{eq:phase_upper_bound}
	\sigma_{\Phi_k}(\mathbf{q}_0,\mathbf{q}_k) \leq \frac{1}{ \mathcal{A}}\sqrt{\frac{1-\mathcal{A}^2}{2n_L}}, \forall k \in \{1,2\}.
\end{equation}
Based on (\ref{eq:height_error}), (\ref{eq:perpendicular_baseline}), and (\ref{eq:phase_upper_bound}), we construct an upper bound on the height error $\sigma_{h_k}$ of interferometric pair $(U_0,U_k)$ as follows:
\begin{equation}\label{eq:objective_function_upperbound}
	\sigma_{h_k}(\mathbf{q}_0,\mathbf{q}_k)\leq  \frac{ \lambda r_0(\mathbf{q}_0 )\sin(\theta_0)  }{2 \pi b_{\bot,k}(\mathbf{q}_k)\mathcal{A}}\sqrt{\frac{1-\mathcal{A}^2}{2n_L}}, \forall k \in \{1,2\}.
\end{equation}
The proof is concluded by inserting the upper bound  (\ref{eq:objective_function_upperbound}) into the expression for the height error of the  final \ac{dem}, given in (\ref{eq:final_height_error}).
\else
\fi
\ifarxiv
\section{Proof of Proposition \ref{prop:baseline_decorrelation}} \label{app:baseline_decorrelation}
Let function $f(x)=\frac{1}{B_p} \left[ \frac{2+B_p}{1+x} -\frac{2-B_p}{1+\frac{1}{x}} \right]$ such that $f(\mathcal{X}(\mathbf{q}_k))$ represents the baseline decorrelation. Note that $\forall \mathbf{q}_k, 0\leq \mathcal{X}(\mathbf{q}_k)\leq 2$, therefore, we focus on function $f$ in the domain $[0 ,2]$. It can be shown that function $f$ is a decreasing and invertible function such that:
\begin{equation}
	f^{-1}(x)=h(x)= \frac{B_p x -2-B_p}{B_p-2-xB_p}.
\end{equation}  
Therefore, constraint $\mathrm{C7}$ is equivalent to the following constraint: 
\begin{equation}\label{eq:constraint_C7_equivalent}
	\mathrm{C7}: f(\mathcal{X}(\mathbf{q}_k)) \geq \gamma_{\rm Rg}^{\rm min} \iff \mathcal{X}(\mathbf{q}_k) \leq h(\gamma_{\rm Rg}^{\rm min} ).
\end{equation}
Based on (\ref{eq:baseline_decorrelation_X}) and (\ref{eq:constraint_C7_equivalent}), constraint $\mathrm{C7}$ is equivalent to the following constraints: 
\begin{equation}\label{eq:app_C7_equivalence}
 \begin{dcases}
		\mathrm{C7a}:\sin(\theta_k(\mathbf{q}_k))\geq \frac{2-h(\gamma_{\rm Rg}^{\rm min})}{h(\gamma_{\rm Rg}^{\rm min})}\sin(\theta_0), \text{if } \mathbf{q}_k \in \mathcal{I}_{a}^k,\\
		 \mathrm{C7b}:\sin(\theta_k(\mathbf{q}_k)) \leq \frac{ h(\gamma_{\rm Rg}^{\rm min})}{2-h(\gamma_{\rm Rg}^{\rm min})}\sin(\theta_0),  \text{if } \mathbf{q}_k \in \mathcal{I}_{b}^k.
	\end{dcases}
\end{equation}
Note that constraints $\mathrm{C7a}$ and $\mathrm{C7b}$ are in general non-convex, however, knowing that  $0 \leq \theta_k(\mathbf{q}_k)\leq \frac{\pi}{2}, \forall \mathbf{q}_k, \forall k$, then we can apply increasing functions $\arcsin$ and  $\tan$ to (\ref{eq:app_C7_equivalence}), which results in the following equivalent convex constraints: 
\ifonecolumn
\begin{equation}
	\begin{dcases}
		\mathrm{C7a}:	\frac{x_t-x_k}{z_k}\geq\tan\left(  \arcsin \left( \frac{\left(2-h(\gamma_{\rm Rg}^{\rm min})\right)\sin(\theta_0)}{h(\gamma_{\rm Rg}^{\rm min})}\right)\right), \text{if } \mathbf{q}_k \in \mathcal{I}_{a}^k, {\forall k,}\\
			\mathrm{C7b}: \frac{x_t-x_k}{z_k}\leq \tan\left( \arcsin \left( \frac{ h(\gamma_{\rm Rg}^{\rm min})\sin(\theta_0)}{2-h(\gamma_{\rm Rg}^{\rm min})}\right)\right),  \text{if } \mathbf{q}_k \in \mathcal{I}_{b}^k, {\forall k.}
	\end{dcases}
\end{equation}
\else
\begin{equation}
	\begin{dcases}
		\mathrm{C7a}:	\frac{x_t-x_k}{z_k}\geq\tan\left(  \sin^{-1} \left( \frac{\left(2-h(\gamma_{\rm Rg}^{\rm min})\right)\sin(\theta_0)}{h(\gamma_{\rm Rg}^{\rm min})}\right)\right), \\ \text{if } \mathbf{q}_k \in \mathcal{I}_{a}^k,\\
		\mathrm{C7b}: \frac{x_t-x_k}{z_k}\leq \tan\left( \sin^{-1} \left( \frac{ h(\gamma_{\rm Rg}^{\rm min})\sin(\theta_0)}{2-h(\gamma_{\rm Rg}^{\rm min})}\right)\right),  \\ \text{if } \mathbf{q}_k \in \mathcal{I}_{b}^k.
	\end{dcases}
\end{equation}
\fi
\else
\fi

\bibliographystyle{IEEEtran}
\bibliography{biblio}

\end{document}